\documentclass[11pt,reqno]{amsart}

\title{Discretizations of isometries}
\author{Pierre-Antoine Guihéneuf}
\address{Universit\'e Paris-Sud \\ Universidade Federal Fluminense, Niterói}
\email{pierre-antoine.guiheneuf@math.u-psud.fr}

\usepackage[latin1]{inputenc}
\usepackage[french,english]{babel}
\usepackage[T1]{fontenc}
\usepackage{amsfonts}
\usepackage{amsmath}
\usepackage{amssymb}
\usepackage{stmaryrd}
\usepackage{amsthm}
\usepackage{enumerate}
\usepackage{pgf,tikz}
\usetikzlibrary{decorations.pathreplacing, shapes.multipart, arrows, matrix}
\usepackage[margin=0pt]{caption}
\usepackage[pdftex,colorlinks=true,linkcolor=blue,citecolor=blue,urlcolor=blue]{hyperref}

\newtheorem{lemme}{Lemma}
\newtheorem{theoreme}[lemme]{Theorem}
\newtheorem{prop}[lemme]{Proposition}

\newtheorem{theorem}{Theorem}

\theoremstyle{definition}
\newtheorem{definition}[lemme]{Definition}
\newtheorem{notation}[lemme]{Notation}

\theoremstyle{remark}
\newtheorem{rem}[lemme]{Remark}

\newcommand{\N}{\mathbf{N}}

\newcommand{\R}{\mathbf{R}}

\newcommand{\Q}{\mathbf{Q}}
\newcommand{\Z}{\mathbf{Z}}
\newcommand{\varep}{\varepsilon}

\newcommand{\Leb}{\operatorname{Leb}}

\newcommand{\Prb}{\mathcal{P}}

\newcommand{\card}{\operatorname{Card}}
\newcommand{\dist}{\operatorname{dist}}

\newcommand{\1}{\mathbf 1}

\makeatletter
\newenvironment{abstracts}{%
  \ifx\maketitle\relax
    \ClassWarning{\@classname}{Abstract should precede
      \protect\maketitle\space in AMS document classes; reported}%
  \fi
  \global\setbox\abstractbox=\vtop \bgroup
    \normalfont\Small
    \list{}{\labelwidth\z@
      \leftmargin3pc \rightmargin\leftmargin
      \listparindent\normalparindent \itemindent\z@
      \parsep\z@ \@plus\p@
      
      \itemsep\bigskipamount
    }%
}{%
  \endlist\egroup
  \ifx\@setabstract\relax \@setabstracta \fi
}

\newcommand{\abstractin}[1]{%
  \otherlanguage{#1}%
  \item[\hskip\labelsep\scshape\abstractname.]%
}
\makeatother

\begin{document}

\begin{abstracts}
\abstractin{english}
This paper deals with the dynamics of discretizations of isometries of $\R^n$, and more precisely the density of the successive images of $\Z^n$ by the discretizations of a generic sequence of isometries. We show that this density tends to 0 as the time goes to infinity. Thus, in general, all the information of a numerical image will be lost by applying many times a naive algorithm of rotation.

\abstractin{french}
On étudie ici la dynamique des discrétisations d'isométries de $\R^n$, et plus précisément la densité des images successives du réseau $\Z^n$ par les discrétisations d'une suite générique d'isométries. On montre que cette densité tend vers 0 quand le temps tend vers l'infini. Ainsi, en général, on va perdre toute l'information contenue dans une image numérique en appliquant plusieurs fois un algorithme naif de rotation.
\end{abstracts}

\selectlanguage{english}

\maketitle

\section[Introduction]{Introduction}

In this paper, we consider the dynamical behaviour of the discretizations of (linear) isometries of a real and finite dimensional vector space $\R^n$. The goal is to understand how it is possible to rotate a numerical image (made of pixels) with the smallest loss of quality as possible.For example, in Figure~\ref{PoincareRot}, we have applied 40 successive random rotations to a $500\times 684$ pixels picture, using a consumer software. These discretized rotations induce a very strong blur in the resulting image, thus a big loss of information.

\begin{figure}[t]
\begin{center}
\begin{minipage}[c]{.3\linewidth}
	\includegraphics[width=\linewidth]{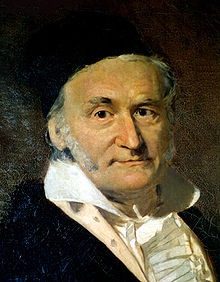}
\end{minipage}\hspace{50pt}
\begin{minipage}[c]{.3\linewidth}
	\includegraphics[width=\linewidth]{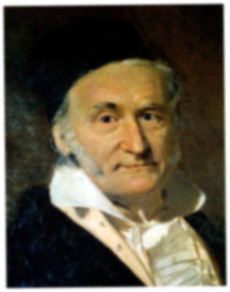}
\end{minipage}
\caption[40 successive rotations of an image]{Original image (left) of size $220\times 282$ and 10 successive random rotations of this image (right), obtained with the software \emph{Gimp} (linear interpolation algorithm).}\label{PoincareRot}
\end{center}	
\end{figure}

Here, we consider the simplest algorithm that can be used to perform a discrete rotation: discretizing the rotation. More precisely, if $x\in\Z^2$ is a integer point (representing a pixel), then the image of $x$ by the discretization of a rotation $R$ will be the integer point which is the closest of $R(x)$. More precisely, in the general case of isometries we will use the following definition of a discretization.

\begin{definition}\label{DefDiscrLin}
We define the projection $p : \R\to\Z$\index{$p$} such that for every $x\in\R$, $p(x)$ is the unique integer $k\in\Z$ such that $k-1/2 < x \le k + 1/2$. This projection induces the map\index{$\pi$}
\[\begin{array}{rrcl}
\pi : & \R^n & \longmapsto & \Z^n\\
 & (x_i)_{1\le i\le n} & \longmapsto & \big(p(x_i)\big)_{1\le i\le n}
\end{array}\]
which is an Euclidean projection on the lattice $\Z^n$. For $P\in O_n(\R)$, we denote by $\widehat P$ the \emph{discretization}\index{$\widehat P$} of $P$, defined by 
\[\begin{array}{rrcl}
\widehat P : & \Z^n & \longrightarrow & \Z^n\\
 & x & \longmapsto & \pi(Px).
\end{array}\]
\end{definition}

We will measure the loss of information induced by the action of discretizing by the \emph{density} of the image set. More precisely, given a sequence $(P_k)_{k\ge 1}$ of linear isometries of $\R^n$, we will look for the density of the set $\Gamma_k = (\widehat{P_k}\circ\cdots\circ\widehat{P_1})(\Z^n)$.

\begin{definition}\label{DefTaux}
For $A_1,\cdots,A_k \in O_n(\R)$, the \emph{rate of injectivity in time $k$} of this sequence is the quantity
\[\tau^k(P_1,\cdots,P_k) = \limsup_{R\to +\infty} \frac{\card \big((\widehat{P_k}\circ\cdots\circ\widehat{P_1})(\Z^n)\cap [B_R]\big)}{\card [B_R]}\in]0,1],\]
where $B_R$ denotes the infinite ball of radius $R$ and $[B_R]$ the set of integral points (i.e. with integer coordinates) inside $B_R$. For an infinite sequence $(P_k)_{k\ge 1}$ of isometries, as the previous quantity is decreasing, we can define the \emph{asymptotic rate of injectivity}\index{$\tau^\infty$}
\[\tau^\infty\big((P_k)_{k\ge 1}\big) = \lim_{k\to +\infty}\tau^k(P_1,\cdots,P_k)\in[0,1].\]
\end{definition}

An example of the sets $\Gamma_k$ for a random draw of isometries $P_m$ is presented on Figure~\ref{ImagesSuitesRotations}. In particular, we observe that the density of these sets seems to decrease when $k$ goes to infinity: the images get whiter and whiter.

\begin{figure}[ht]
\begin{minipage}[c]{.33\linewidth}
	\includegraphics[width=\linewidth, trim = 1.5cm .5cm 1.5cm .5cm,clip]{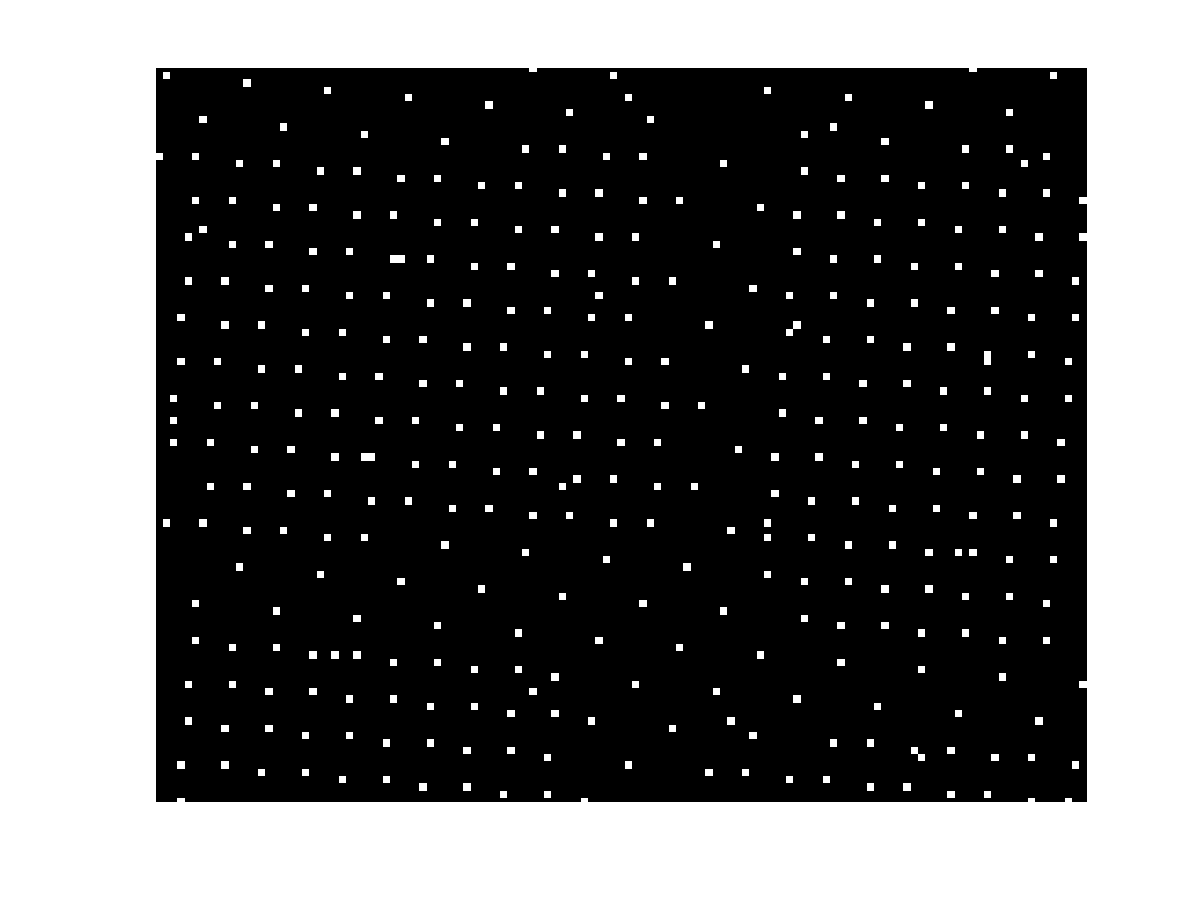}
\end{minipage}\hfill
\begin{minipage}[c]{.33\linewidth}
	\includegraphics[width=\linewidth, trim = 1.5cm .5cm 1.5cm .5cm,clip]{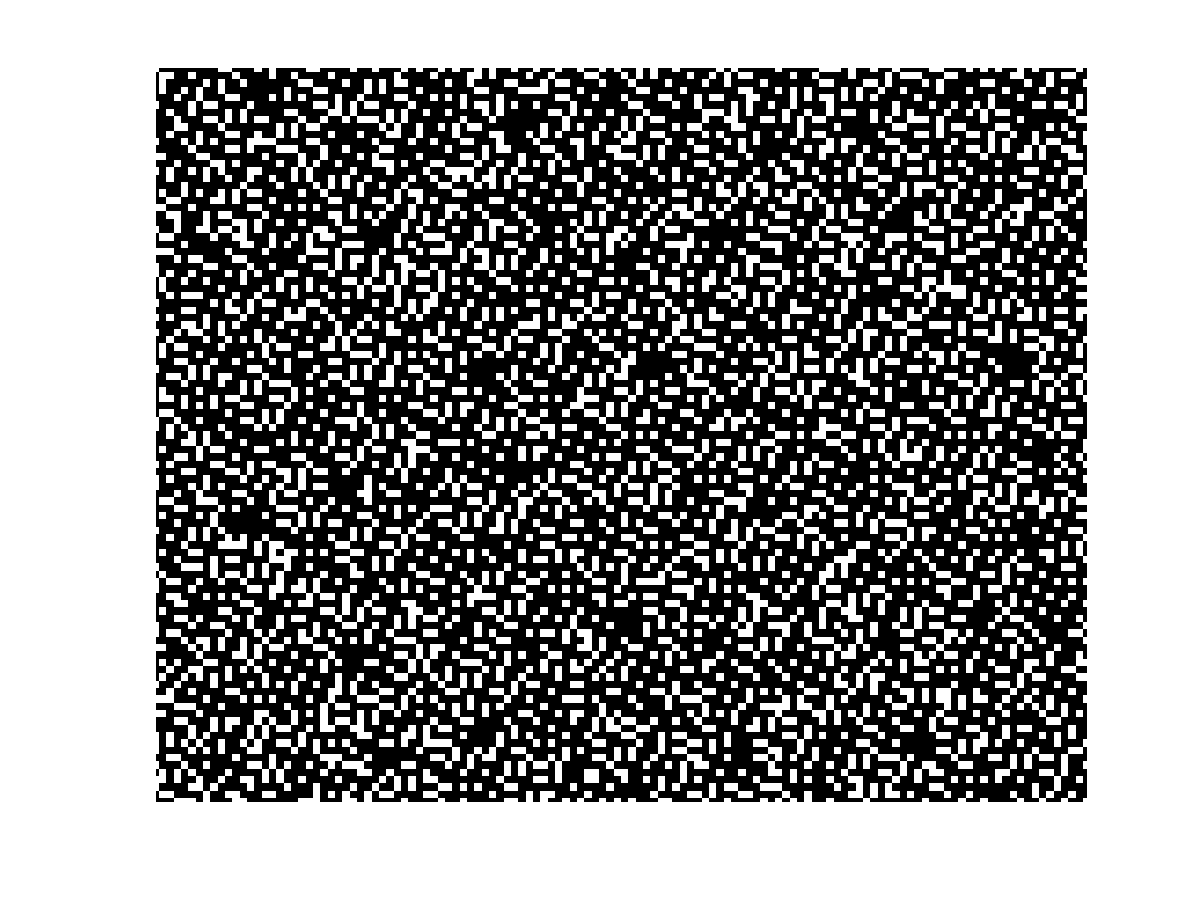}
\end{minipage}\hfill
\begin{minipage}[c]{.33\linewidth}
	\includegraphics[width=\linewidth, trim = 1.5cm .5cm 1.5cm .5cm,clip]{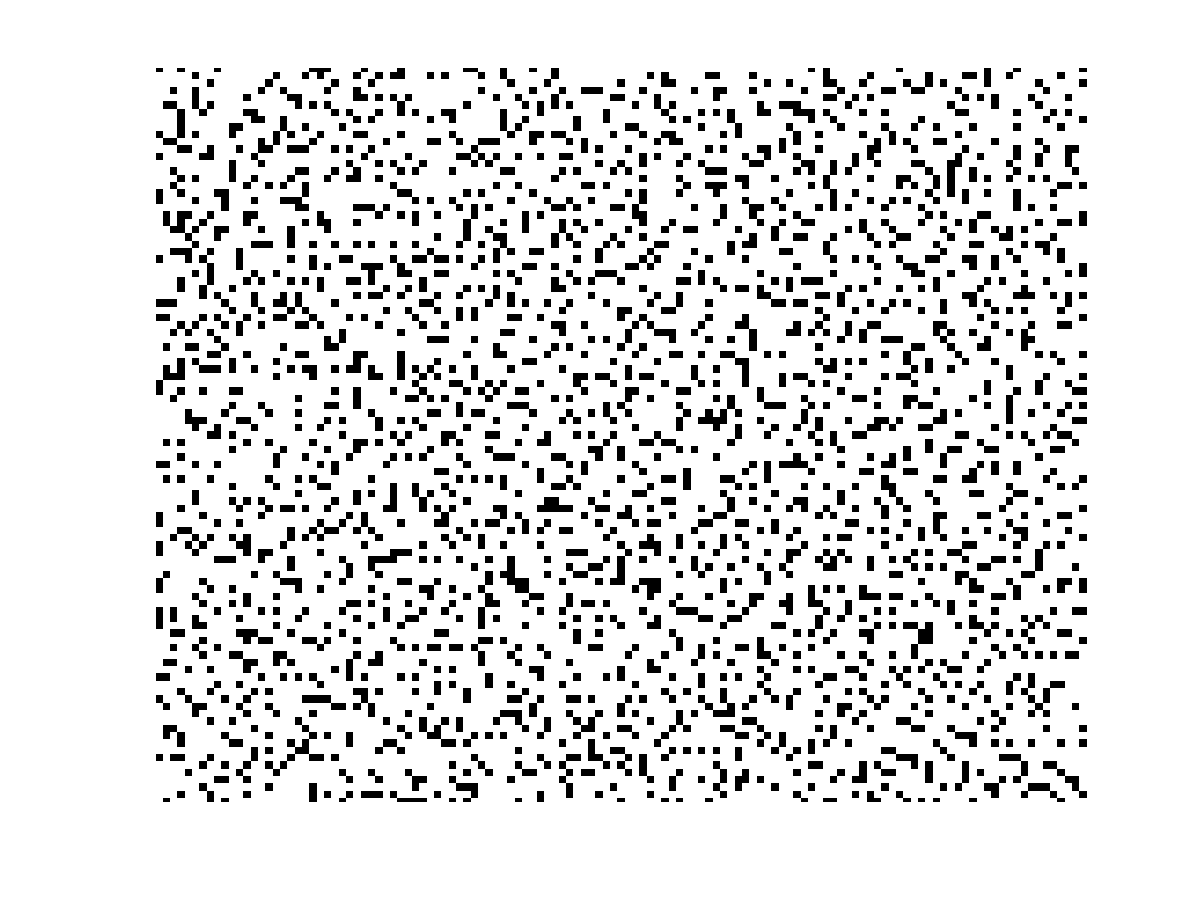}
\end{minipage}
\caption[Successive images of $\Z^2$ by discretizations of random rotations]{Successive images of $\Z^2$ by discretizations of random rotations, a point is black if it belongs to $(\widehat{R_{\theta_k}}\circ\cdots\circ\widehat{R_{\theta_1}})(\Z^2)$, where the $\theta_i$ are chosen uniformly randomly in $[0,2\pi]$. From left to right and top to bottom, $k=2,\, 5,\, 50$.}\label{ImagesSuitesRotations}
\end{figure}

This phenomenon is confirmed when we plot the density of the intersection between these image sets $\Gamma_k$ and a big ball of $\R^n$ (see Figure~\ref{TauxSuiteRotations}): this density seems to tend to 0 as the time $k$ goes to infinity.

\begin{figure}[t]
\begin{center}
\includegraphics[width=.5\linewidth, trim = 1.2cm 1.2cm 1.2cm 1cm, clip]{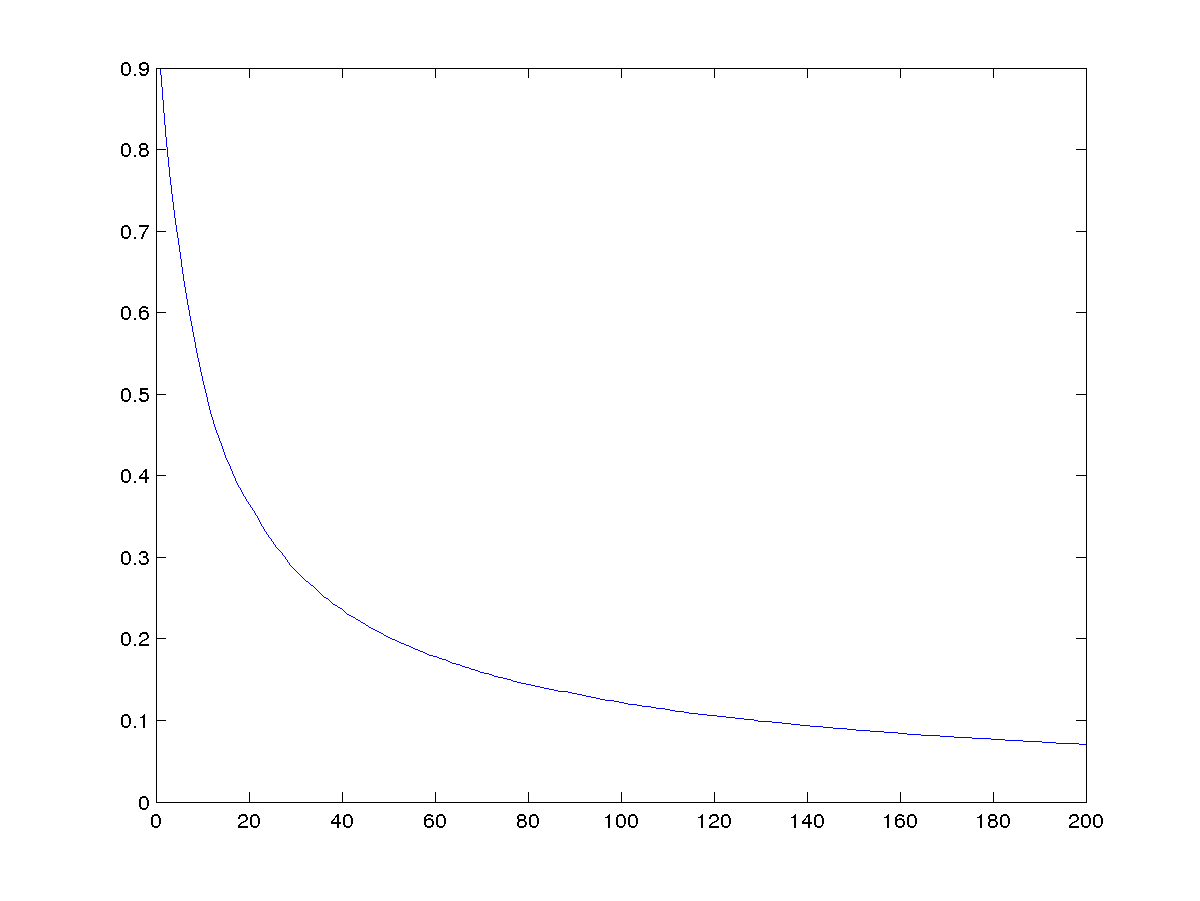}
\caption[Expectation of the rate of injectivity of a random sequences of rotations]{Expectation of the rate of injectivity of a random sequences of rotations: the graphic represents the mean of the rate of injectivity $\tau^k(R_{\theta_k},\cdots,R_{\theta_1})$ depending on $k$, $1\le k\le 200$, for 50 random draws of sequences of angles $(\theta_i)_i$, with each $\theta_i$ chosen independently and uniformly in $[0,2\pi]$. Note that the behaviour is not exponential.}\label{TauxSuiteRotations}
\end{center}
\end{figure}

We would like to explain theoretically this phenomenon. Of course, if we take $P_m = \operatorname{Id}$, then we will have $\Gamma_k = \Z^n$ and the rates of injectivity will be equal to 0. To avoid this kind of ``exceptional cases'', we will study the asymptotic rate of injectivity of a \emph{generic} sequence of matrices of $O_n(\R)$, in the following sense.

\begin{definition}\label{DefTopoSL}
We fix once for all a norm $\|\cdot\|$ on $M_n(\R)$. For any sequence $(P_k)_{k\ge 1}$ of matrices of $O_n(\R)$, we set\index{$\|(P_k)_k\|$}
\[\|(P_k)_k\|_\infty = \sup_{k\ge 1} \|P_k\|.\]
In other words, we consider the space $\ell^\infty(O_n(\R))$ of uniformly bounded sequences of linear isometries endowed with this natural metric.
\end{definition}

This metric is complete, thus there is a good notion of genericity on the set of linear isometries: a set $\mathcal U\subset (O_n(\R))^\N$ is \emph{generic} if it is a countable intersection of open and dense subsets of $\ell^\infty(O_n(\R))$. The main theorem of this paper studies the asymptotic rate of injectivity in this context.

\begin{theorem}\label{DD}
Let $(P_k)_{k\ge 1}$ be a generic sequence of matrices of $O_n(\R)$. Then $\tau^\infty((P_k)_k) = 0$.
\end{theorem}

The proof of this theorem will even show that for every $\varep>0$, there exists an open and dense subset of $\ell^\infty(O_n(\R))$ on which $\tau^\infty$ is smaller than $\varep$. This theorem expresses that for ``most of'' the sequences of isometries, the loss of information is total. Thus, for a generic sequence of rotations, we will not be able to avoid the blur observed in Figure~\ref{PoincareRot}.

Note that we do not know what is the rate of injectivity of a sequence of isometries made of independent identically distributed random draws (for example with respect to the Haar measure on $O_n(\R)$).  
\medskip

The proof of Theorem~\ref{DD} will be the occasion to study the structure of the image sets $\Gamma_k = (\widehat{P_k}\circ\cdots\circ\widehat{P_1})(\Z^n)$. It appears that there is a kind of ``regularity at infinity'' in $\Gamma_k$. More precisely, this set is an \emph{almost periodic pattern}: roughly speaking, for $R$ large enough, the set $\Gamma_k \cap B_R$ determines the whole set $\Gamma_k$ up to an error of density smaller than $\varep$ (see Definition~\ref{DefAlmPer}). We prove that the image of an almost periodic pattern by the discretization of a linear map is still an almost periodic pattern (Theorem~\ref{imgquasi}); thus, the sets $\Gamma_k$ are almost periodic patterns.

The idea of the proof of Theorem~\ref{DD} is to take advantage of the fact that for a generic sequence of isometries, we have a kind of independence of the coefficients of the matrices. Thus, for a generic isometry $P\in O_n(\R)$, the set $P(\Z^n)$ is uniformly distributed modulo $\Z^n$. We then remark that the local pattern of the image set $\widehat P(\Z^n)$ around $\widehat P(x)$ is only determined by $P$ and the the remainder of $Px$ modulo $\Z^n$: the global behaviour of $\widehat P(\Z^n)$ is coded by the quotient $\R^n/\Z^n$. This somehow reduces the study to a local problem.

As a first application of this remark, we state that the rate of injectivity in time 1 can be seen as the area of an intersection of cubes (Proposition~\ref{FormTau1}). This observation is one of the two keys of the proof of Theorem~\ref{DD}, the second one being the study of the action of the discretizations $\widehat P$ on the frequencies of differences $\rho_{\Gamma_k}(v) = D\big(({\Gamma_k}-v)\cap{\Gamma_k}\big)$. Indeed, if there exists a set $\Gamma'\subset\Gamma$ of positive density, together with a vector $v$ such that for every $x\in\Gamma'$, we have $\widehat P(x) = \widehat P(x+v)$, then we will have $D(\widehat P(\Gamma)) \le D(\widehat P)-D(\Gamma')$. This study of the frequencies of differences will include a Minkowski-type theorem for almost-periodic patterns (Theorem~\ref{MinkAlm}).
\medskip

\label{BibliLin} The particular problem of the discretization of linear maps has been quite little studied. To our knowledge, what has been made in this direction has been initiated by image processing. One wants to avoid phenomenons like loss of information (due to the fact that discretizations of linear maps are not injective) or aliasing (the apparition of undesirable periodic patterns in the image, due for example to a resonance between a periodic pattern in the image and the discretized map). To our knowledge, the existing studies are mostly interested in the linear maps with \emph{rational coefficients} (see for example \cite{Jacob25}, \cite{MR1382839} or \cite{MR1832794}), including the specific case of \emph{rotations} (see for example \cite{A1996_1075}, \cite{nouvel:tel-00444088}, \cite{thibault:tel-00596947}, \cite{MR1782038}). These works mainly focus on the local behaviour of the images of $\Z^2$ by discretizations of linear maps: given a radius $R$, what pattern can follow the intersection of this set with any ball of radius $R$? What is the number of such patterns, what are their frequencies? Are they complex (in a sense to define) or not? Are these maps bijections? In particular, the thesis \cite{nouvel:tel-00444088} of B.~Nouvel gives a characterization of the angles for which the discrete rotation is a bijection (such angles are countable and accumulate only on $0$). Our result complements that of B.~Nouvel: one the one hand it expresses that a generic sequence of discretizations is far from being a bijection, and on the other hand this remains true in any dimension.

Note that Theorem~\ref{DD} will be generalized to the case of matrices of determinant 1 in \cite{Gui15a}, with more sophisticated techniques (see also \cite{Guih-These}).

\section{Almost periodic sets}\label{ChapAlm}

In this section, we introduce the basic notions that we will use during the study of discretizations of isometries of $\R^n$.

We fix once for all an integer $n\ge 1$. We will denote by $\llbracket a, b \rrbracket$\index{$\llbracket\cdot\rrbracket$} the integer segment $[a,b]\cap\Z$. In this part, every ball will be taken with respect to the infinite norm; in particular, for $x = (x_1,\cdots,x_n)$, we will have\index{$B(x,R)$}
\[B(x,R) = B_\infty(x,R) = \big\{y=(y_1,\cdots,y_n)\in\R^n\mid \forall i\in \llbracket 1, n\rrbracket, |x_i-y_i|<R\big\}.\]
We will also denote $B_R = B(0,R)$\index{$B_R$}. Finally, we will denote by $\lfloor x \rfloor$\index{$\lfloor \cdot \rfloor$} the biggest integer that is smaller than $x$ and $\lceil x \rceil$\index{$\lceil \cdot \rceil$} the smallest integer that is bigger than $x$. For a set $B\subset \R^n$, we will denote $[B] = B\cap \Z^n$.\index{$B$@$[B]$}% Finally, we will denote $\1_{x\in E}$ the indicator function of the event $x\in E$.

\subsection[Almost periodic patterns]{Almost periodic patterns: definitions and first properties}

In this subsection, we define the notion of almost periodic pattern and prove that these sets possess a uniform density.

\begin{definition}
Let $\Gamma$ be a subset of $\R^n$.
\begin{itemize}
\item We say that $\Gamma$ is \emph{relatively dense} if there exists $R_\Gamma>0$ such that each ball with radius at least $R_\Gamma$ contains at least one point of $\Gamma$.
\item We say that $\Gamma$ is a \emph{uniformly discrete} if there exists $r_\Gamma>0$ such that each ball with radius at most $r_\Gamma$ contains at most one point of $\Gamma$.
\end{itemize}
The set $\Gamma$ is called a \emph{Delone} set if it is both relatively dense and uniformly discrete.
\end{definition}

\begin{definition}
For a discrete set $\Gamma\subset \R^n$ and $R\ge 1$, we define the uniform $R$-density:\index{$D_R^+$}
\[D_R^+(\Gamma) = \sup_{x\in\R^n} \frac{\card\big(B(x,R)\cap \Gamma\big)}{\card\big(B(x,R)\cap\Z^n\big)},\]
and the uniform upper density:\index{$D^+$}
\[D^+(\Gamma) = \limsup_{R\to +\infty} D_R^+(\Gamma).\]
\end{definition}

Remark that if $\Gamma\subset \R^n$ is Delone for the parameters $r_\Gamma$ and $R_\Gamma$, then its upper density satisfies:
\[\frac{1}{(2R_\Gamma+1)^n} \le D^+(\Gamma) \le \frac{1}{(2r_\Gamma+1)^n}.\]

We can now define the notion of almost periodic pattern that we will use throughout this paper. Roughly speaking, an almost periodic pattern $\Gamma$ is a set for which there exists a relatively dense set of translations of $\Gamma$, where a vector $v$ is a translation of $\Gamma$ if $\Gamma-v$ is equal to $\Gamma$ up to a set of upper density smaller than $\varep$. More precisely, we state the following definition.

\begin{definition}\label{DefAlmPer}\index{$\mathcal N_\varep$}
A Delone set $\Gamma\subset\Z^n$ is an \emph{almost periodic pattern} if for every $\varep>0$, there exists $R_\varep>0$ and a relatively dense set $\mathcal N_\varep$, called the \emph{set of $\varep$-translations} of $\Gamma$, such that
\begin{equation}\label{EqAlmPer}
\forall R\ge R_\varep,\  \forall v\in\mathcal N_\varep,\  D_R^+\big( (\Gamma+v)\Delta \Gamma \big) <\varep.
\end{equation}
\end{definition}

Of course, every lattice, or every finite union of translates of a given lattice, is an almost periodic pattern. We will see in next subsection a large class of examples of almost periodic patterns: images of $\Z^n$ by discretizations of linear maps.
\medskip

We end this introduction to almost periodic patterns by stating that the notion of almost periodic pattern is invariant under discretizations of linear isometries: the image of an almost periodic pattern by the discretization of a linear isometry is still an almost periodic pattern.

\begin{theoreme}\label{imgquasi}
Let $\Gamma\subset\Z^n$ be an almost periodic pattern and $P\in O_n(\R)$. Then the set $\widehat P(\Gamma)$ is an almost periodic pattern.
\end{theoreme}

This implies that, given a sequence $(P_k)_{k\ge 1}$ of isometries of $\R^n$, the successive images $(\widehat{P_k}\circ\cdots\circ \widehat{P_1})(\Z^n)$ are almost periodic patterns. See Figure~\ref{ImagesSuitesRotations} for an example of the successive images of $\Z^2$ by a random sequence of bounded matrices of $O_2(\R)$. The proof of Theorem~\ref{imgquasi} will be done in Appendix~\ref{TrouVentre}. Examples of sets $\widehat P(\Z^2)$ for various rotations $P$ can be found in Figure~\ref{ImagesRotations}, where the almost periodicity is patent.

\begin{figure}[ht]
\begin{minipage}[c]{.33\linewidth}
	\includegraphics[width=\linewidth, trim = 1.5cm .5cm 1.5cm .5cm,clip]{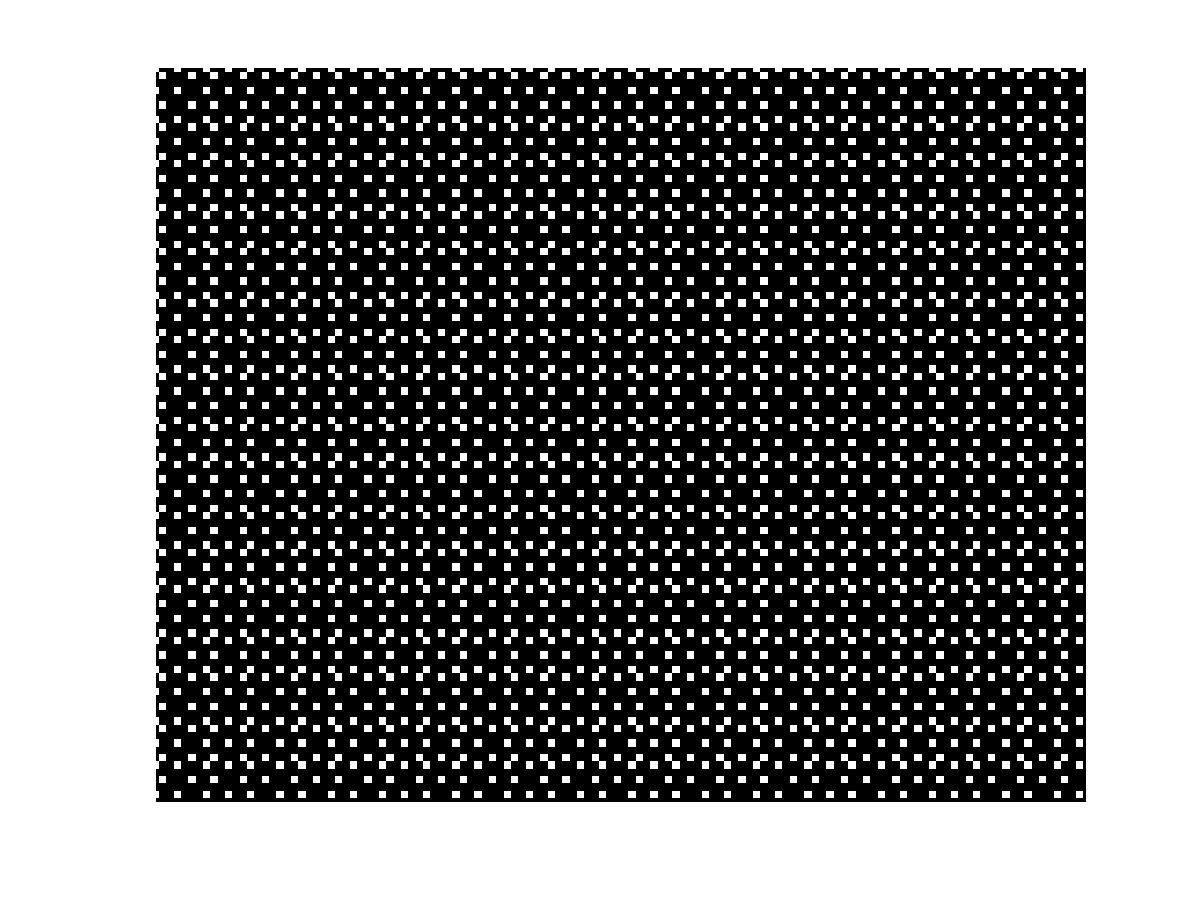}
\end{minipage}\hfill
\begin{minipage}[c]{.33\linewidth}
	\includegraphics[width=\linewidth, trim = 1.5cm .5cm 1.5cm .5cm,clip]{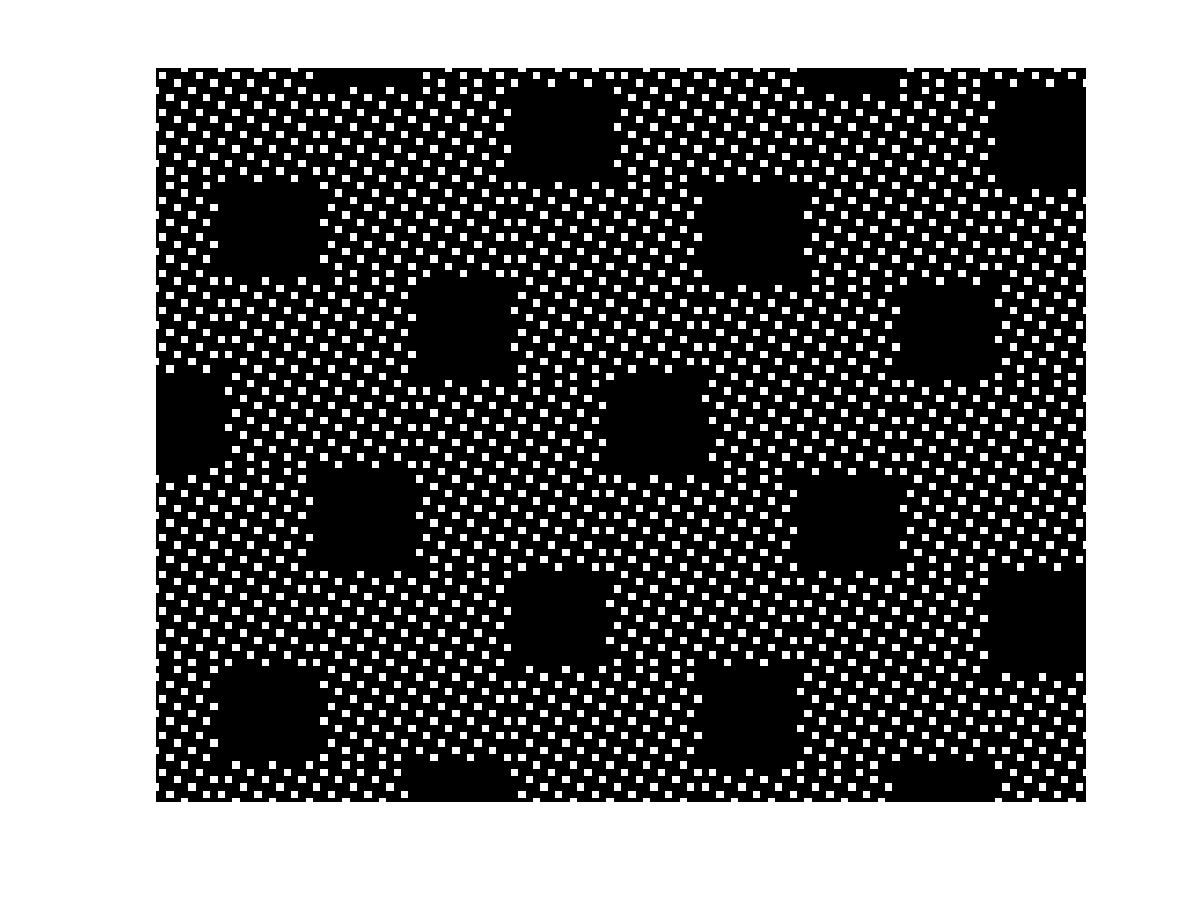}
\end{minipage}\hfill
\begin{minipage}[c]{.33\linewidth}
	\includegraphics[width=\linewidth, trim = 1.5cm .5cm 1.5cm .5cm,clip]{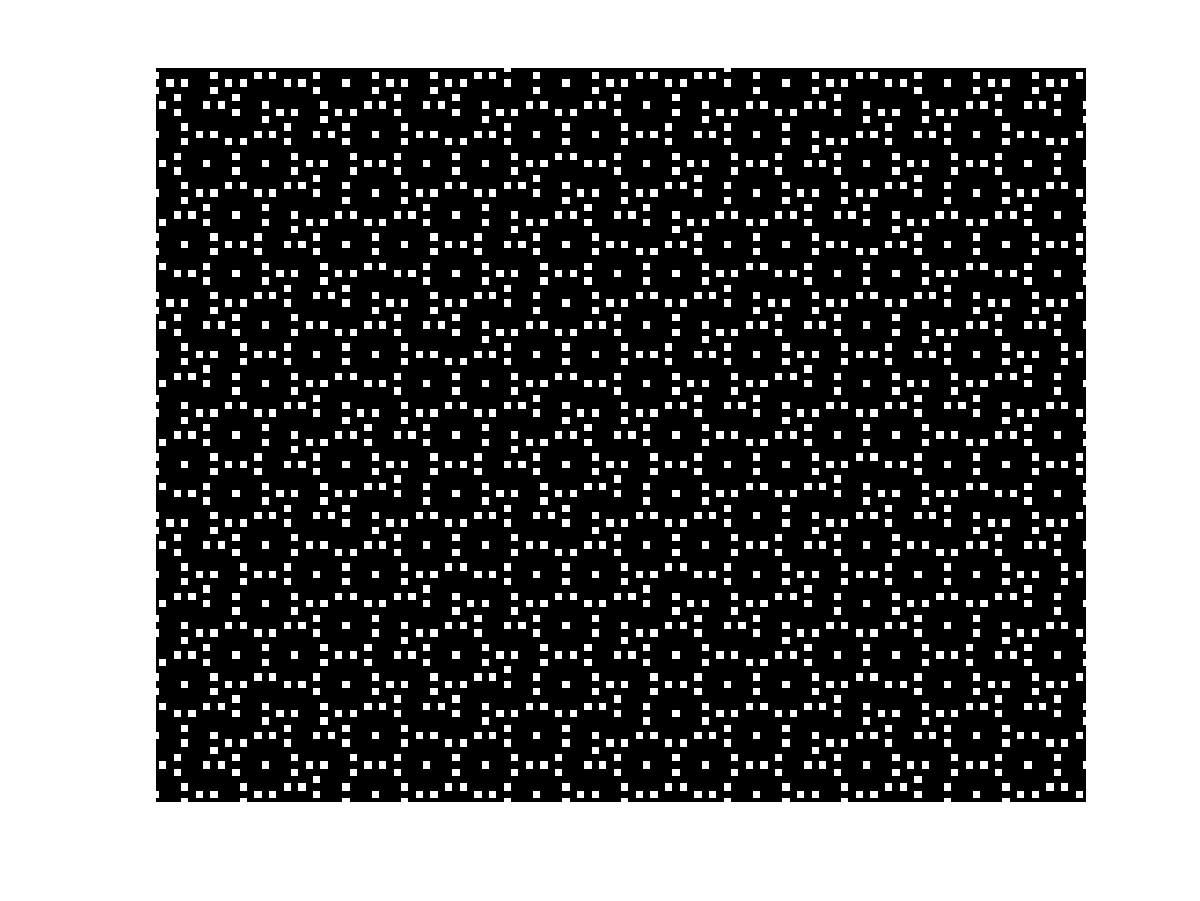}
\end{minipage}

\caption[Images of $\Z^2$ by discretizations of rotations]{Images of $\Z^2$ by discretizations of rotations, a point is black if it belongs to the image of $\Z^2$ by the discretization of the rotation. From left to right and top to bottom, angles $\pi/4$, $\pi/5$ and $\pi/6$.}\label{ImagesRotations}
\end{figure}

\subsection{Differences in almost periodic patterns}

We will need to understand how behave the differences in an almost periodic pattern $\Gamma$, i.e. the vectors $x-y$ with $x,y\in\Gamma$. In fact, we will study the frequency of appearance of these differences.

\begin{definition}\label{DefDiff}
For $v\in\Z^n$, we set\index{$\rho_\Gamma$}
\[\rho_\Gamma(v) = \frac{D\{x\in\Gamma\mid x+v\in\Gamma\}}{D(\Gamma)} = \frac{D\big(\Gamma\cap(\Gamma-v)\big)}{D(\Gamma)} \in [0,1]\]
the \emph{frequency} of the difference $v$ in the almost periodic pattern $\Gamma$.
\end{definition}

Studying frequencies of differences allows to focus on the global behaviour of an almost periodic set. The function $\rho_\Gamma$ is itself almost periodic in the sense given by H. Bohr (see \cite{MR1555192}).

\begin{definition}
Let $f : \Z^n\to \R$. Denoting by $T_v$\index{$T_v$} the translation of vector $v$, we say that $f$ is \emph{Bohr almost periodic} (also called \emph{uniformly almost periodic}) if for every $\varep>0$, the set
\[\mathcal N_\varep = \big\{v\in\Z^n\mid \|f - f\circ T_v \|_\infty<\varep \big\},\]
is relatively dense.
\end{definition}

If $f: \Z^n\to \R$ is a Bohr almost periodic function, then it possesses a \emph{mean}\index{$\mathcal M$} $\mathcal M(f)$ (see for example the historical paper of H. Bohr \cite[Satz VIII]{MR1555192}), which satisfies: for every $\varep>0$, there exists $R_0>0$ such that for every $R\ge R_0$ and every $x\in\R^n$, we have
\[\left|\mathcal M(f) - \frac{1}{\card[B(x,R)]}\sum_{v\in [B(x,R)]} f(v) \right| <\varep.\]

The fact that $\rho_\Gamma$ is Bohr almost periodic is straightforward.

\begin{lemme}
If $\Gamma$ is an almost periodic pattern, then the function $\rho_\Gamma$ is Bohr almost periodic.
\end{lemme}

In fact, we can compute precisely the mean of $\rho(\Gamma)$.

\begin{prop}\label{IntRho}
If $\Gamma$ is an almost periodic pattern, then we have 
\[\mathcal M (\rho_\Gamma) = D(\Gamma).\]
\end{prop}

The proof of this proposition will be done in Appendix~\ref{AppenTech}.

We now state a Minkowski-type theorem for the map $\rho_\Gamma$. To begin with, we recall the classical Minkowski theorem (see for example the book \cite{MR893813}).

\begin{theoreme}[Minkowski]\label{Minkowski}
Let $\Lambda$ be a lattice of $\R^n$, $k\in\N$ and $S\subset\R^n$ be a centrally symmetric convex body. If $\Leb(S/2) > k \operatorname{covol}(\Lambda)$, then $S$ contains at least $2k$ distinct points of $\Lambda\setminus\{0\}$.
\end{theoreme}

\begin{theoreme}\label{MinkAlm}
Let $\Gamma\subset \Z^n$ be an almost periodic pattern of density $D(\Gamma)$. Let $S$ be a centrally symmetric body, with $\Leb(S)> 4^nk$. If for every $v\in S\cap \Z^n$, we have $\rho_\Gamma(v)<\rho_0$, then
\[\rho_0\ge\frac{1}{k}\left(1-\frac{1}{D(\Gamma)(2k+1)}\right).\]
In particular, if $k\ge \frac{1}{D(\Gamma)}$, then there exists $x\in C\cap \Z^n$ such that $\rho_\Gamma(x)\ge \frac{D(\Gamma)}{2}$.
\end{theoreme}

\begin{proof}[Proof of Theorem~\ref{MinkAlm}]
Minkowski theorem (Theorem~\ref{Minkowski}) asserts that $S/2$ contains at least $2k+1$ distinct points of $\Z^n$, denoted by $u^i$. By the hypothesis on the value of $\rho_\Gamma$ on $S$, and because the set of differences of $S/2$ is included in $S$, we know that the density of $(\Gamma + u^i)\cap (\Gamma+u^j)$ is smaller than $\rho_0 D(\Gamma)$. Thus,
\begin{align*}
D\Big(\bigcup_i (\Gamma +u^i)\Big) & \ge \sum_i D(\Gamma) - \sum_{i<j} D\big((\Gamma + u^i)\cap (\Gamma+u^j)\big)\\
  & \ge (2k+1)D(\Gamma) - \frac{2k(2k+1)}{2}\rho_0 D(\Gamma).
\end{align*}
The theorem the follows from the fact that the left member of this inequality is smaller than 1.
\end{proof}

\section{Rate of injectivity of isometries}\label{Souris}

We now focus in more detail on the rate of injectivity of a sequence of isometries (see Definition~\ref{DefTaux}).

\subsection{A geometric viewpoint on the rate of injectivity}\label{ptgeom}

In this subsection, we present a geometric construction to compute the rate of injectivity of a generic matrix, and some applications of it.

Let $P\in O_n(\R)$ and $\Lambda = P(\Z^n)$. The density of $\pi(\Lambda)$ is the proportion of $x\in\Z^n$ belonging to $\pi(\Lambda)$; in other words the proportion of $x\in\Z^n$ such that there exists $\lambda\in \Lambda$ whose distance to $x$ (for $\|\cdot\|_\infty$) is smaller than $1/2$. remark that this property only depends on the value of $x$ modulo $\Lambda$. If we consider the union\index{$U$}
\[U = \bigcup_{\lambda\in\Lambda} B(\lambda,1/2)\]
of balls of radius $1/2$ centred on the points of $\Lambda$ (see Figure~\ref{RateRot}), then $x\in\pi(\Lambda)$ if and only if $x\in U\cap\Z^n$. So, if we set $\nu$ the measure of repartition of the $x\in\Z^n$ modulo $\Lambda$, that is\index{$\nu$}
\[\nu = \lim_{R\to+\infty} \frac{1}{\card(B_R\cap \Z^n)} \sum_{x\in B_R\cap \Z^n} \delta_{\operatorname{pr}_{\R^n/\Lambda}(x)},\]
then we obtain the following formula.

\begin{prop}\label{FormTau1}
For every $P\in O_n(\R)$ (we identify $U$ with its projection of $\R^n/\Lambda$),
\[\tau(P) = D\big(\pi(\Lambda)\big) = \nu\big(\operatorname{pr}_{\R^n/\Lambda}(U)\big).\]
\end{prop}

An even more simple formula holds when the matrix $P$ is totally irrational.

\begin{definition}\label{TotIrrat}
We say that a matrix $P\in O_n(\R)$ is \emph{totally irrational} if the image $P(\Z^n)$ is equidistributed\footnote{It is equivalent to require that it is dense instead of equidistributed.} modulo $\Z^n$; in particular, this is true when the coefficients of $P$ form a $\Q$-free family.
\end{definition}

If the matrix $P$ is totally irrational, then the measure $\nu$ is the uniform measure. Thus, if $\mathcal D$ is a fundamental domain of $\R^n/\Lambda$, then $\tau(P)$ is the area of $\mathcal D \cap U$.

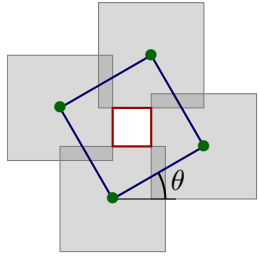
\begin{figure}[ht]
\begin{center}
\begin{tikzpicture}[scale=1.4]
\clip (-1,-.6) rectangle (1.5,2.2);
\foreach\i in {0,...,1}{
\foreach\j in {0,...,1}{
\fill[color=gray,opacity = .3] (-.5+.866*\i-.5*\j,-.5+.5*\i+.866*\j) rectangle (.5+.866*\i-.5*\j,.5+.5*\i+.866*\j);
\draw[color=gray] (-.5+.866*\i-.5*\j,-.5+.5*\i+.866*\j) rectangle (.5+.866*\i-.5*\j,.5+.5*\i+.866*\j);
}}
\draw (0,0) -- (.6,0);
\draw[thick] (.5,0) arc (0:30:.5);
\draw (.62,.18) node {$\theta$};
\draw[color=red!60!black,thick] (0,.866) -- (0,.5) -- (.366,.5) -- (.366,.866) -- cycle;
\draw[color=blue!40!black,thick] (0,0) -- (.866,.5) -- (.866-.5,.866+.5) -- (-.5,.866) -- cycle;
\foreach\i in {0,...,1}{
\foreach\j in {0,...,1}{
\draw[color=green!40!black] (.866*\i-.5*\j,.5*\i+.866*\j) node {$\bullet$};
}}
\end{tikzpicture}
\caption[Mean rate of injectivity of a rotation]{Computation of the mean rate of injectivity of a rotation of $\R^2$: it is equal to $1$ minus the area of the interior of the red square.}\label{RateRot}
\end{center}
\end{figure}

With the same kind of arguments, we easily obtain a formula for $\rho_{\widehat P(\Z^n)}(v)$ (the frequency of the difference $v$ in $\widehat P(\Z^n)$, see Definition~\ref{DefDiff}).

\begin{prop}\label{ActionDiffGeom}
If $P\in GL_n(\R)$ is totally irrational, then for every $v\in\Z^n$,
\[\rho_{\widehat P(\Z^n)}(v) = \Leb\big( B(v,1/2) \cap U\big).\]
\end{prop}

\begin{proof}[Sketch of proof of Proposition~\ref{ActionDiffGeom}]
We want to know which proportion of points $x\in\Gamma = \widehat P(\Z^n)$ are such that $x+v$ also belongs to $\Gamma$. But modulo $\Lambda = P(\Z^n)$, $x$ belongs to $\Gamma$ if and only if $x\in B(0,1/2)$. Similarly, $x+v$ belongs to $\Gamma$ if and only if $x\in B(-v,1/2)$. Thus, by equirepartition, $\rho_{\widehat P(\Z^n)}(v)$ is equal to the area of $B(v,1/2) \cap U$.
\end{proof}

From Proposition~\ref{FormTau1}, we deduce the continuity of $\overline\tau$. More precisely, $\overline\tau$ is continuous and piecewise polynomial of degree smaller than $n$; moreover $\tau$ coincides with a continuous function on a generic subset of $O_n(\R)$.

It also allows to compute simply the mean rate of injectivity of some examples of matrices: for $\theta\in[0,\pi/2]$, the mean rate of injectivity of a rotation of $\R^2$ of angle $\theta$ is (see Figure \ref{RateRot}).
\[\overline\tau(R_\theta) = 1-(\cos(\theta)+\sin(\theta)-1)^2.\]

\subsection{Diffusion process}

In this paragraph, we study the action of a discretization of a matrix on the set of differences of an almost periodic pattern $\Gamma$; more precisely, we study the link between the functions $\rho_\Gamma$ and $\rho_{\widehat P(\Gamma)}$.

For $u\in\R^n$, we define the function $\varphi_u$, which is a ``weighted projection'' of $u$ on $\Z^n$.

\begin{definition}
Given $u\in\R^n$, the function $\varphi_u = \Z^n\to [0,1]$ is defined by
\[\varphi_u (v) = \left\{\begin{array}{ll}
0 & \ \text{if}\ d_\infty(u,v)\ge 1\\
\prod_{i=1}^n (1-|u_i+v_i|) & \ \text{if}\ d_\infty(u,v)< 1.
\end{array}\right.\]
\end{definition}

\begin{figure}
\begin{minipage}[b]{.4\linewidth}
\begin{center}
\begin{tikzpicture}[scale=.85]
\draw (0,0) rectangle (3,3);
\draw[thick, color=blue!80!black] (1.3,0.9) rectangle (3,3);
\draw (1.3,0.9) node {$\times$};
\draw (1.3,0.9) node[above left] {$u$};
\draw[color=green!40!black] (0,0) node {$\bullet$};
\draw (0,0) node[below right] {$v$};
\draw[color=green!40!black] (0,3) node {$\bullet$};
\draw (0,3) node[above right] {$v+(0,1)$};
\draw[color=green!40!black] (3,0) node {$\bullet$};
\draw (3,0) node[below right] {$v+(1,0)$};
\draw[color=green!40!black] (3,3) node {$\bullet$};
\draw (3,3) node[above right] {$v+(1,1)$};
\draw[->] (1.35,.75) to[bend right] (2.9,.1);
\draw[->] (1.35,1.05) to[bend left] (2.9,2.9);
\draw[->] (1.25,.75) to[bend left] (.1,.1);
\draw[->] (1.25,1.05) to[bend right] (.1,2.9);
\draw[->, color=red!60!black] (.1,.1) to (1.2,0.8);
\draw[color=blue!80!black] (1.3,0) -- (1.3,3);
\draw[color=blue!80!black] (0,.9) -- (3,.9);
\end{tikzpicture}
\end{center}
\caption[The function $\varphi_u$]{The function $\varphi_u$ in dimension 2: its value on one vertex of the square is equal to the area of the opposite rectangle; in particular, $\varphi_u(v)$ is the area of the rectangle with the vertices $u$ and $v+(1,1)$ (in bold).}\label{IxMa}
\end{minipage}\hfill
\begin{minipage}[b]{.55\linewidth}
\begin{center}
\begin{tikzpicture}[scale=.8]
\draw[thick, color=blue!80!black] (1.7,2.1) rectangle (0,0);
\draw[color=blue!80!black] (1.7,0) -- (1.7,3);
\draw[color=blue!80!black] (0,2.1) -- (3,2.1);
\draw[color=blue!80!black] (0,0) rectangle (3,3);
\draw (1.7,2.1) node {$\times$};
\draw (1.7,2.1) node[above left] {$u'$};

\draw[dashed] (0,0) rectangle (6,6);
\draw[dashed] (3,0) -- (3,6);
\draw[dashed] (0,3) -- (6,3);
%\draw (1.5,0) -- (1.5,6);
%\draw (0,1.5) -- (6,1.5);
%\draw (4.5,0) -- (4.5,6);
%\draw (0,4.5) -- (6,4.5);
\draw (1.5,1.5) node[below left] {\small$(0,0)$};
\draw (4.5,1.5) node[below right] {\small$(1,0)$};
\draw (1.5,4.5) node[above left] {\small$(0,1)$};
\draw (4.5,4.5) node[above right] {\small$(1,1)$};
\draw[->, color=red!60!black] (1.8,2.2) to (2.9,2.9);
\draw[color=green!40!black] (1.5,1.5) node {$\bullet$};
\draw[color=green!40!black] (1.5,4.5) node {$\bullet$};
\draw[color=green!40!black] (4.5,1.5) node {$\bullet$};
\draw[color=green!40!black] (4.5,4.5) node {$\bullet$};
\end{tikzpicture}
\end{center}
\caption[Proof of Proposition~\ref{ActionDiff}]{The red vector is equal to that of Figure~\ref{IxMa} for $u=Pv$. If $Px$ belongs to the bottom left rectangle, then $\pi(Px+Pv) = y\in\Z^2$; if $Px$ belongs to the top left rectangle, then $\pi(Px+Pv) = y+(0,1)$ etc.}\label{IxMa2}
\end{minipage}
\end{figure}

In particular, the function $\varphi_u$ satisfies $\sum_{v\in\Z^n} \varphi_u(v) = 1$, and is supported by the vertices of the integral unit cube\footnote{An integral cube has vertices with integer coordinates and its faces parallel to the canonical hyperplanes of $\R^n$.} that contains\footnote{More precisely, the support of $\varphi_u$ is the smallest integral unit cube of dimension $n'\le n$ which contains $u$.} $u$. Figure~\ref{IxMa} gives a geometric interpretation of this function $\varphi_u$.

The following property asserts that the discretization $\widehat P$ acts ``smoothly'' on the frequency of differences. In particular, when $D(\Gamma) = D(\widehat P \Gamma)$, the function $\rho_{\widehat P \Gamma}$ is obtained from the function $\rho_\Gamma$ by applying a linear operator $\mathcal A$, acting on each Dirac function $\delta_v$ such that $\mathcal A \delta_u(v) = \varphi_{P(u)}(v)$. Roughly speaking, to compute $\mathcal A \delta_v$, we take $\delta_{Pv}$ and apply a diffusion process. In the other case where $D(\widehat P \Gamma) < D(\Gamma)$, we only have inequalities involving the operator $\mathcal A$ to compute the function $\rho_{\widehat P\Gamma}$.

\begin{prop}\label{ActionDiff}
Let $\Gamma\subset \Z^n$ be an almost periodic pattern and $P\in O_n(\R)$ be a generic matrix.
\begin{enumerate}[(i)]
\item If $D(\widehat P(\Gamma)) = D(\Gamma)$, then for every $u\in\Z^n$,
\[\rho_{\widehat P(\Gamma)}(u) = \sum_{v\in\Z^n} \varphi_{P(v)} (u) \rho_\Gamma(v).\]
\item In the general case, for every $u\in\Z^n$, we have
\[\frac{D(\Gamma)}{D(\widehat P(\Gamma))}\sup_{v\in\Z^n} \varphi_{P(v)} (u) \rho_\Gamma(v) \le \rho_{\widehat P(\Gamma)}(u) \le \frac{D(\Gamma)}{D(\widehat P(\Gamma))}\sum_{v\in\Z^n} \varphi_{P(v)} (u) \rho_\Gamma(v).\]
\end{enumerate}
\end{prop}

\begin{proof}[Proof of Proposition \ref{ActionDiff}]
We begin by proving the first point of the proposition. Suppose that $P\in O_n(\R)$ is generic and that $D(\widehat P(\Gamma)) = D(\Gamma)$. Let $x\in \Gamma\cap (\Gamma-v)$. We consider the projection $y'$ of $y=Px$, and the projection $u'$ of $u=Pv$, on the fundamental domain $]-1/2,1/2]^n$ of $\R^n/\Z^n$. We have
\[P(x+v) = \pi(Px) + \pi(Pv) + y' + u'.\]
Suppose that $y'$ belongs to the parallelepiped whose vertices are $(-1/2,\cdots,-1/2)$ and $u'$ (in bold in Figure~\ref{IxMa2}), then $y'+u'\in [-1/2,1/2[^n$. Thus, $\pi(P(x+v)) = \pi(Px) + \pi(Pv)$. The same kind of results holds for the other parallelepipeds whose vertices are $u'$ and one vertex of $[-1/2,1/2[^n$.

We set $\Gamma = \widehat P(\Z^n)$. The genericity of $P$ ensures that for every $v\in\Z^n$, the set $\Gamma\cap (\Gamma-v)$, which has density $D(\Gamma)\rho_\Gamma(v)$ (by definition of $\rho_\Gamma$), is equidistributed modulo $\Z^n$ (by Lemma~\ref{passifacil}). Thus, the points $x'$ are equidistributed modulo $\Z^n$. In particular, the difference $v$ will spread into the differences which are the support of the function $\varphi_{Pv}$, and each of them will occur with a frequency given by $\varphi_{Pv} (x) \rho_\Gamma(v)$. The hypothesis about the fact that the density of the sets does not decrease imply that the contributions of each difference of $\Gamma$ to the differences of $\widehat P (\Gamma)$ add.

In the general case, the contributions to each difference of $\Gamma$ may overlap. However, applying the argument of the previous case, we can easily prove the second part of the proposition.
\end{proof}

\begin{rem}\label{RemActionDiff}
We also remark that:
\begin{enumerate}[(i)]
\item the density strictly decreases (that is, $D(\widehat P(\Gamma)) < D(\Gamma)$) if and only if there exists $v_0\in\Z^n$ such that $\rho_\Gamma(v_0)>0$ and $\|Pv_0\|_\infty <1$;
\item if there exists $v_0\in\Z^n$ such that 
\[\sum_{v\in\Z^n} \varphi_{P(v)} (v_0) \rho_\Gamma(v_0)>1,\]
then the density strictly decreases by at least $\sum_{v\in\Z^n} \varphi_{P(v)} (v_0) \rho_\Gamma(v_0)- 1$.	
%\item we can compute which differences will go to the difference 0, that is, the differences $u\in\Z^n$ such that there exists $x\in\Gamma\cap(\Gamma-u)$ such that $\widehat A(x) = \widehat A(x+u)$ (that will make the rate of injectivity decrease). The set of such differences $u$ is $A^{-1}(B(0,1))\cap \Z^n$ (recall that $B(0,1)$ is an infinite ball). By Minkowski theorem (Theorem~\ref{Minkowski}), this set contains a non zero vector when $A$ is generic. More generally, the set of differences $u\in \Gamma\cap(\Gamma-u)$ such that $\widehat A(x) + v = \widehat A(x+u)$ is $A^{-1}(B(v,1))\cap \Z^n$. Iterating this process, it is possible to compute which differences will go to 0 in time $t$.
\end{enumerate}
\end{rem}

\subsection{Rate of injectivity of a generic sequence of isometries}

We now come to the proof of the main theorem of this paper (Theorem~\ref{DD}). We will begin by applying the Minkowski theorem for almost periodic patterns (Theorem~\ref{MinkAlm}), which gives \emph{one} nonzero difference whose frequency is positive. The rest of the proof of Theorem~\ref{DD} consists in using again an argument of equidistribution. More precisely, we apply successively the following lemma, which asserts that given an almost periodic pattern $\Gamma$ of density $D_0$, a sequence of isometries and $\delta>0$, then, perturbing each isometry of at most $\delta$ if necessary, we can make the density of the $k_0$-th image of $\Gamma$ smaller than $\lambda_0 D_0$, with $k_0$ and $\lambda_0$ depending only on $D_0$ and $\delta$. The proof of this lemma involves the study of the action of the discretizations on differences made in Proposition~\ref{ActionDiff}

\begin{lemme}\label{EstimPerteRot}
Let $(P_k)_{k\ge 1}$ be a sequence of matrices of $O_n(\R)$ and $\Gamma\subset\Z^n$ an almost periodic pattern. Given $\delta>0$ and $D>0$ such that $D(\Gamma)\ge D$, there exists $k_0 = k_0(D)$ (decreasing in $D$), $\lambda_0 = \lambda_0(D,\delta)<1$ (decreasing in $D$ and in $\delta$), and a sequence $(Q_k)_{k\ge 1}$ of totally irrational matrices of $O_n(\R)$, such that $\|P_k - Q_k\|\le \delta$ for every $k\ge 1$ and
\[D\big((\widehat{Q_{k_0}}\circ\dots\circ \widehat{Q_1})(\Gamma)\big)< \lambda_0 D(\Gamma).\]
\end{lemme}

We begin by proving that this lemma implies Theorem~\ref{DD}.

\begin{proof}[Proof of Theorem \ref{DD}]
Suppose that Lemma~\ref{EstimPerteRot} is true. Let $\tau_0\in ]0,1[$ and $\delta>0$. We want to prove that we can perturb the sequence $(P_k)_k$ into a sequence $(Q_k)_k$ of isometries, which is $\delta$-close to $(P_k)_k$ and is such that its asymptotic rate is smaller than $\tau_0$ (and that this remains true on a whole neighbourhood of these matrices).

Thus, we can suppose that $\tau^\infty((P_k)_k)>\tau_0$. We apply Lemma~\ref{EstimPerteRot} to obtain the parameters $k_0 = k_0(\tau_0/2)$ (because $k_0(D)$ is decreasing in $D$) and $\lambda_0 = \lambda_0(\tau_0/2,\delta)$ (because $\lambda_0(D,\delta)$ is decreasing in $D$). Applying the lemma $\ell$ times, this gives a sequence $(Q_k)_k$ of isometries, which is $\delta$-close to $(P_k)_k$, such that, as long as $\tau^{\ell k_0}(Q_0,\cdots,Q_{\ell k_0})> \tau_0/2$, we have $\tau^{\ell k_0}(Q_1,\cdots,Q_{\ell k_0})<\lambda_0^\ell D(\Z^n)$. But for $\ell$ large enough, $\lambda_0^\ell<\tau_0$, which proves the theorem.
\end{proof}

\begin{proof}[Proof of Lemma \ref{EstimPerteRot}]
The idea of the proof is the following. Firstly, we apply the Minkowski-type theorem for almost periodic patterns (Theorem~\ref{MinkAlm}) to find a uniform constant $C>0$ and a point $u_0\in\Z^n\setminus\{0\}$ whose norm is ``not too big'', such that $\rho_\Gamma (u_0) > C D(\Gamma)$. Then, we apply Proposition \ref{ActionDiff} to prove that the difference $u_0$ in $\Gamma$ eventually goes to 0; that is, that there exists $k_0\in\N^*$ and an almost periodic pattern $\widetilde\Gamma$ of positive density (that can be computed) such that there exists a sequence $(Q_k)_k$ of isometries, with $\|Q_i-P_i\|\le\delta$, such that for every $x\in \widetilde\Gamma$,
\[(\widehat{Q_{k_0}}\circ\dots\circ \widehat{Q_1})(x) = (\widehat{Q_{k_0}}\circ\dots\circ \widehat{Q_1})(x+u_0).\]
This makes the density of the $k_0$-th image of $\Gamma$ decrease:
\[D\big((\widehat{Q_{k_0}}\circ\dots\circ \widehat{Q_1})(\Gamma)\big) \le D(\Gamma)-D(\widetilde \Gamma);\]
a precise estimation of the density of $\widetilde\Gamma$ will then prove the lemma.
\medskip

We begin by applying the Minkowski-like theorem for almost periodic patterns (Theorem \ref{MinkAlm}) to a \emph{Euclidean} ball $B'_R$\index{$B'_R$} such that (recall that $[B]$ denotes the set of integer points inside $B$)
\begin{equation}\label{superEq}
\Leb(B'_R) = V_n R^n = 4^n\left\lfloor\frac{1}{D(\Gamma)}\right\rfloor,
\end{equation}
where $V_n$ denotes the measure of the unit ball on $\R^n$. Then, Theorem \ref{MinkAlm} says that there exists $u_0\in B'_R\cap \Z^n\setminus\{0\}$ such that
\begin{equation}\label{EqDefU0}
\rho_\Gamma(u_0)\ge \frac{D(\Gamma)}{2}.
\end{equation}

We now perturb each matrix $P_k$ into a totally irrational matrix $Q_k$ such that for every point $x\in [B'_R]\setminus\{0\}$, the point $Q_k (x)$ is far away from the lattice $\Z^n$. More precisely, as the set of matrices $Q\in O_n(\R)$ such that $Q([B'_R]) \cap \Z^n \neq \{0\}$ is finite, there exists a constant $d_0(R,\delta)$ such that for every $P\in O_n(\R)$, there exists $Q\in O_n(\R)$ such that $\|P-Q\|\le\delta$ and for every $x\in [B'_R]\setminus\{0\}$, we have $d_\infty(Q(x),\Z^n)> d_0(R,\delta)$. Applying Lemma~\ref{passifacil} (which states that if the sequence $(Q_k)_k$ is generic, then the matrices $Q_k$ are ``non resonant''), we build a sequence $(Q_k)_{k\ge 1}$ of totally irrational\footnote{See Definition~\ref{TotIrrat}.} matrices of $O_n(\R)$ such that for every $k\in\N^*$, we have:
\begin{itemize}
\item $\|P_k-Q_k\|\le\delta$;
\item for every $x\in [B'_R]\setminus\{0\}$, we have $d_\infty(Q_k(x),\Z^n)> d_0(R,\delta)$;
\item the set $(Q_k\circ \widehat{Q_{k-1}} \circ \cdots\circ \widehat{Q_1})(\Gamma)$ is equidistributed modulo $\Z^n$.
\end{itemize}

We then consider the difference $u_0$ (given by Equation~\eqref{EqDefU0}). We denote by $\lfloor P \rfloor (u)$ the point of the smallest integer cube of dimension $n'\le n$ that contains $P(u)$ which has the smallest Euclidean norm (that is, the point of the support of $\varphi_{P(u)}$ with the smallest Euclidean norm). In particular, if $P(u)\notin \Z^n$, then $\|\lfloor P\rfloor (u)\|_2 < \|P(u)\|_2$ (where $\|\cdot\|_2$ is the Euclidean norm). Then, the point (ii) of Proposition \ref{ActionDiff} shows that 
\begin{align*}
\rho_{\widehat{Q_1}(\Gamma)}(\lfloor Q_1 \rfloor (u_0)) & \ge \frac{D(\Gamma)}{D(\widehat{Q_1}(\Gamma))} \varphi_{Q_1(\lfloor Q_1 \rfloor (u_0))} (u_0) \rho_\Gamma(u_0)\\
     & \ge \frac{\big(d_0(R,\delta)\big)^n}{2}D(\Gamma),
\end{align*}
(applying Equation~\eqref{EqDefU0}) and so on, for every $k\in\N^*$,
\[\rho_{(\widehat{Q_k}\circ\cdots\circ \widehat{Q_1})(\Gamma)}\big((\lfloor Q_{k} \rfloor \circ \cdots \circ \lfloor Q_1 \rfloor )(u_0)\big) \ge \left(\frac{\big(d_0(R,\delta)\big)^n}{2}\right)^k D(\Gamma).\]

We then remark that the sequence 
\[\big\|(\lfloor Q_{k} \rfloor \circ \cdots \circ \lfloor Q_1 \rfloor )(u_0)\big\|_2\]
is decreasing and can only take a finite number of values (it lies in $\sqrt\Z$). Then, there exists $k_0\le R^2$ such that
\[\big(\lfloor Q_{k_0} \rfloor \circ \cdots \circ \lfloor Q_1 \rfloor \big) (u_0) = 0;\]
in particular, by Equation \eqref{superEq}, we have
\[k_0\le \left(\frac{4^n}{V_n}\left\lfloor\frac{1}{D(\Gamma)}\right\rfloor\right)^{2/n}.\]
Then, point (ii) of Remark \ref{RemActionDiff} applied to $v_0 = 0$ implies that the density of the image set satisfies
\[D\big((\widehat{Q_k}\circ\cdots\circ \widehat{Q_1})(\Gamma)\big) \le \left(1-\left(\frac{\big(d_0(R,\delta)\big)^n}{2}\right)^{k_0}\right) D(\Gamma).\]
The conclusions of the lemma are obtained by setting $\lambda_0 = 1-\left(\frac{(d_0(R,\delta))^n}{2}\right)^{k_0}$.
\end{proof}

\appendix

\section{Technical lemmas}\label{AppenTech}

Let us begin by the proof of Proposition \ref{IntRho}.

\begin{proof}[Proof of Proposition \ref{IntRho}]
This proof lies primarily in an inversion of limits.

Let $\varep>0$. As $\Gamma$ is an almost periodic pattern, there exists $R_0>0$ such that for every $R\ge R_0$ and every $x\in\R^n$, we have
\begin{equation}\label{eqDens}
\left|D(\Gamma) - \frac{\Gamma \cap [B(x,R)]}{\card[B_R]}\right|\le \varep.
\end{equation}

So, we choose $R\ge R_0$, $x\in\Z^n$ and compute
\begin{align*}
\frac{1}{\card[B_R]} & \sum_{v\in[B(x,R)]} \rho_\Gamma(v) = \frac{1}{\card[B_R]}\sum_{v\in[B(x,R)]} \frac{D\big((\Gamma-v)\cap \Gamma\big)}{D(\Gamma)}\\
       = & \frac{1}{\card[B_R]}\sum_{v\in[B(x,R)]} \lim_{R'\to +\infty}\frac{1}{\card[B_{R'}]}\sum_{y\in[B_{R'}]} \frac{\1_{y\in\Gamma-v} \1_{y\in\Gamma}}{D(\Gamma)}\\
       = & \frac{1}{D(\Gamma)}\lim_{R'\to +\infty}\frac{1}{\card[B_{R'}]} \sum_{y\in[B_{R'}]} \1_{y\in\Gamma}\frac{1}{\card[B_R]}\sum_{v\in[B(x,R)]} \1_{y\in\Gamma-v}\\
       = & \frac{1}{D(\Gamma)}\underbrace{\lim_{R'\to +\infty}\frac{1}{\card[B_{R'}]} \sum_{y\in[B_{R'}]} \1_{y\in\Gamma}}_{\text{first term}}\underbrace{\frac{1}{\card[B_R]}\sum_{v'\in[B(y+x,R)]} \1_{v'\in\Gamma}}_{\text{second term}}.
\end{align*}
By Equation \eqref{eqDens}, the second term is $\varep$-close to $D(\Gamma)$. Considered independently, the first term is equal to $D(\Gamma)$ (still by Equation \eqref{eqDens}). Thus, we have
\[\left|\frac{1}{\card[B(x,R)]} \sum_{v\in[B(x,R)]} \rho_\Gamma(v) - D(\Gamma)\right|\le \varep,\]
that we wanted to prove.
\end{proof}

We now state an easy lemma which asserts that for $\varep$ small enough, the set of translations $\mathcal N_\varep$ is ``stable under additions with a small number of terms''.

\begin{lemme}\label{arithProg}
Let $\Gamma$ be an almost periodic pattern, $\varep>0$ and $\ell\in\N$. Then if we set $\varep'=\varep/\ell$ and denote by $\mathcal N_{\varep'}$ the set of translations of $\Gamma$ and $R_{\varep'}>0$ the corresponding radius for the parameter $\varep'$, then for every $k\in\llbracket 1,\ell \rrbracket$ and every $v_1,\cdots,v_\ell\in\mathcal N_{\varep'}$, we have
\[\forall R\ge R_{\varep'},\  D_R^+\Big( \big(\Gamma+\sum_{i=1}^\ell v_i\big)\Delta \Gamma \Big) <\varep.\]
\end{lemme}

\begin{proof}[Proof of Lemma \ref{arithProg}]
Let $\Gamma$ be an almost periodic pattern, $\varep>0$, $\ell\in\N$, $R_0>0$ and $\varep'=\varep/\ell$. Then there exists $R_{\varep'}>0$ such that
\[\forall R\ge R_{\varep'},\  \forall v\in\mathcal N_{\varep'},\  D_R^+\big( (\Gamma+v)\Delta \Gamma \big) <\varep'.\]
We then take $1\le k\le\ell$, $v_1,\cdots,v_k\in\mathcal N_{\varep'}$ and compute
\begin{align*}
D_R^+\Big( \big(\Gamma+\sum_{i=1}^k v_i\big)\Delta \Gamma \Big) & \le \sum_{m=1}^k D_R^+\Big( \big(\Gamma+\sum_{i=1}^m v_i\big)\Delta \big(\Gamma+\sum_{i=1}^{m-1} v_i\big) \Big)\\
             & \le \sum_{m=1}^k D_R^+\Big( \big((\Gamma+v_m)\Delta \Gamma\big) + \sum_{i=1}^{m-1} v_i \Big).
\end{align*}
By the invariance under translation of $D_R^+$, we deduce that
\begin{align*}
D_R^+\Big( \big(\Gamma+\sum_{i=1}^k v_i\big)\Delta \Gamma \Big) & \le \sum_{m=1}^k D_R^+ \big((\Gamma+v_m)\Delta \Gamma\big)\\
						 & \le k\varep'.
\end{align*}
As $k\le \ell$, this ends the proof.
\end{proof}

\begin{rem}\label{arithProg2}
In particular, this lemma implies that the set $\mathcal N_\varep$ contains arbitrarily large patches of lattices of $\R^n$: for every almost periodic pattern $\Gamma$, $\varep>0$ and $\ell\in\N$, there exists $\varep'>0$ such that for every $k_i \in \llbracket -\ell,\ell\rrbracket$ and every $v_1,\cdots,v_n\in\mathcal N_{\varep'}$, we have
\[\forall R\ge R_{\varep'},\  D_{R}^+\Big( \big(\Gamma+\sum_{i=1}^n k_iv_i \big)\Delta \Gamma \Big) <\varep.\]
\end{rem}

The second lemma is more technical. It expresses that given an almost periodic pattern $\Gamma$, a generic matrix $A\in O_n(\R)$ is non resonant with respect to $\Gamma$.

\begin{lemme}\label{passifacil}
Let $\Gamma\subset \Z^n$ be an almost periodic pattern with positive uniform density. Then the set of $A\in O_n(\R)$ such that $A(\Gamma)$ is equidistributed modulo $\Z^n$ is generic. More precisely, for every $\varep>0$, there exists an open and dense set of $A\in O_n(\R)$ such that there exists $R_0>0$ such that for every $R>R_0$, the projection on $\R^n/\Z^n$ of the uniform measure on $A(\Gamma\cap B_R)$ is $\varep$-close to Lebesgue measure on $\R^n/\Z^n$.
\end{lemme}

\begin{proof}[Proof of Lemma \ref{passifacil}]
During this proof, we consider a distance $\dist$ on $\Prb(\R^n/\Z^n)$ which is invariant under translations, where $\Prb(\R^n/\Z^n)$ denotes the space of probability Borel measures on $\R^n/\Z^n$ endowed with weak-* topology. We also suppose that this distance satisfies the following convexity inequality: if $\mu, \nu_1,\cdots,\nu_d\in\Prb(\R^n/\Z^n)$, then
\[\dist\left(\mu,\frac{1}{d}\sum_{i=1}^d \nu_i\right) \le \frac{1}{d}\sum_{i=1}^d \dist(\mu_,\nu_i).\]
For the simplicity of the notations, when $\mu$ and $\nu$ have not total mass 1, we will denote by $\dist(\mu,\nu)$ the distance between the normalizations of $\mu$ and $\nu$.

We consider the set $\mathcal U_\varep$ of matrices $A\in GL_n(\R)$ satisfying: there exists $R_0>0$ such that for all $R\ge R_0$,
\[\dist \left(\Leb_{\R^n/\Z^n}, \sum_{x\in B_R\cap \Gamma} \bar\delta_{Ax} \right) <\varep,\]
where $\bar\delta_x$\index{$\bar\delta$} is the Dirac measure of the projection of $x$ on $\R^n/\Z^n$. We show that for every $\varep>0$, $\mathcal U_\varep$ contains an open dense set. Then, the set $\bigcap_{\varep>0}\mathcal U_\varep$ will be a $G_\delta$ dense set made of matrices $A\in GL_n(\R)$ such that $A(\Gamma)$ is well distributed.

\begin{figure}[t]
\begin{center}
\begin{tikzpicture}[scale=.9]
\draw[fill=blue!20!white, thick] (-2.5,-2.5) rectangle (2.5,2.5);
\foreach\i in {-2,...,2}{
\foreach\j in {-2,...,2}{
\draw[fill=gray,opacity=.3] (0.95*\i-0.03*\j-.5,0.06*\i+1.02*\j-.5) rectangle (0.95*\i-0.03*\j+.5,0.06*\i+1.02*\j+.5) ;
\draw (0.95*\i-0.03*\j-.5,0.06*\i+1.02*\j-.5) rectangle (0.95*\i-0.03*\j+.5,0.06*\i+1.02*\j+.5) ;
}}
\draw (0,0) node {$\times$};
\draw[->] (0,0) to (0.95,0.06);
\draw (.5,.1) node[below right] {$v_1$};
\draw[->] (0,0) to (-0.03,1.02);
\draw (.1,.5) node[above left] {$v_2$};
\draw (0,0) node[below left]{$0$};
\end{tikzpicture}
\caption[``Almost tiling'' of $B_{\ell R_0}$]{``Almost tiling'' of $B_{\ell R_0}$ by cubes $B(\sum_{i=1}^n k_i v_i, R_0)$, with $-\ell\le k_i\le\ell$.}\label{AlmTil}
\end{center}
\end{figure}
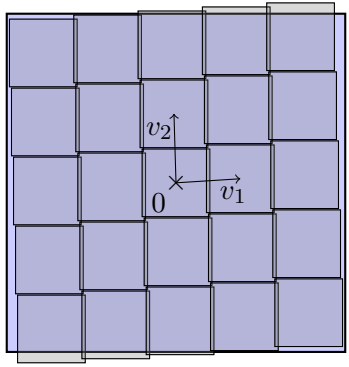

Let $\varep>0$, $\delta>0$, $\ell>0$ and $A\in GL_n(\R)$. We apply Remark~\ref{arithProg2} to obtain a parameter $R_0>0$ and a family $v_1,\cdots,v_n$ of $\varep$-translations of $\Gamma$ such that the family of cubes $\big(B(\sum_{i=1}^n k_i v_i, R_0)\big)_{-\ell\le k_i\le \ell}$ is an ``almost tiling'' of $B_{\ell R_0}$ (in particular, each $v_i$ is close to the vector having $2R_0$ in the $i$-th coordinate and $0$ in the others, see Figure~\ref{AlmTil}):
\begin{enumerate}[(1)]
\item\label{pati1} this collection of cubes fills almost all $B_{\ell R_0}$:
\[\frac{\card\Big(\Gamma \cap \big(\bigcup_{-\ell\le k_i\le \ell}B(\sum_{i=1}^n k_i v_i, R_0) \Delta B_{\ell R_0}\big)\Big)}{\card(\Gamma \cap B_{\ell R_0})} \le \varep;\]
\item\label{pati2} the overlaps of the cubes are not too big: for all collections $(k_i)$ and $(k'_i)$ such that $-\ell \le k_i,k'_i \le \ell$, 
\[\frac{\card\Big(\Gamma \cap \big(B(\sum_{i=1}^n k_i v_i, R_0)\Delta B(\sum_{i=1}^n k'_i v_i, R_0)\big)\Big)}{\card(\Gamma \cap B_{\ell R_0})} \le \varep;\]
\item\label{pati3} the vectors $\sum_{i=1}^n k_i v_i$ are translations for $\Gamma$: for every collection $(k_i)$ such that $-\ell \le k_i \le \ell$, 
\[\frac{\card\Big(\big(\Gamma \Delta (\Gamma-\sum_{i=1}^n k_i v_i)\big)\cap B_{R_0}\Big)}{\card(\Gamma \cap B_{R_0})} \le \varep.\]
\end{enumerate}

Increasing $R_0$ and $\ell$ if necessary, there exists $A'\in GL_n(\R)$ (respectively $SL_n(\R)$, $O_n(\R)$) such that $\|A-A'\|\le\delta$ and that we have
\begin{equation}\label{machinrioa}
\dist\left(\Leb_{\R^n/\Z^n}, \sum_{-\ell\le k_i\le\ell}\bar\delta_{A'(\sum_{i=1}^n k_i v_i)}\right) \le \varep.
\end{equation}
Indeed, if we denote by $\Lambda$ the lattice spanned by the vectors $v_1,\cdots,v_n$, then the set of matrices $A'$ such that $A'\Lambda$ is equidistributed modulo $\Z^n$ is dense in $GL_n(\R)$ (respectively $SL_n(\R)$ and $O_n(\R)$).

Then, we have,
\begin{align*}
\dist\Bigg(\Leb_{\R^n/\Z^n}, & \sum_{\substack{-\ell \le k_i \le \ell\\ x\in \Gamma \cap B(\sum_{i=1}^n k_i v_i, R_0)}} \bar\delta_{A'x}\Bigg) \le\\
  & \dist\Bigg(\Leb_{\R^n/\Z^n}, \sum_{\substack{-\ell \le k_i \le \ell\\ x\in \Gamma \cap B(0, R_0)}} \bar\delta_{A'(\sum_{i=1}^n k_i v_i)+ A'x}\Bigg)\\
  & + \dist\Bigg(\sum_{\substack{-\ell \le k_i \le \ell\\ x\in \Gamma \cap B(0, R_0)}} \bar\delta_{A'(\sum_{i=1}^n k_i v_i)+ A'x}, \sum_{\substack{-\ell \le k_i \le \ell\\ x\in \Gamma \cap B(\sum_{i=1}^n k_i v_i, R_0)}} \bar\delta_{A'x}\Bigg)
\end{align*}
By the property of convexity of $\dist$, the first term is smaller than
\[ \frac{1}{\card\big(\Gamma \cap B(0, R_0)\big)}\sum_{x\in \Gamma \cap B(0, R_0)}\dist\Bigg(\Leb_{\R^n/\Z^n}, \sum_{-\ell \le k_i \le \ell} \bar\delta_{A'(\sum_{i=1}^n k_i v_i)+ A'x}\Bigg);\]
by Equation~\eqref{machinrioa} and the fact that $\dist$ is invariant under translation, this term is smaller than $\varep$. As by hypothesis, the vectors $\sum_{i=1}^n k_i v_i$ are $\varep$-translations of $\Gamma$ (hypothesis \eqref{pati3}), the second term is also smaller than $\varep$. 
Thus, we get 
\[\dist\Bigg(\Leb_{\R^n/\Z^n}, \sum_{\substack{-\ell \le k_i \le \ell\\ x\in \Gamma \cap B(\sum_{i=1}^n k_i v_i, R_0)}} \bar\delta_{A'x}\Bigg) \le 2\varep\]

By the fact that the family of cubes $\big(B(\sum_{i=1}^n k_i v_i, R_0)\big)_{-\ell\le k_i\le \ell}$ is an almost tiling of $B_{\ell R_0}$ (hypotheses \eqref{pati1} and \eqref{pati2}), we get, for every $v\in \R^n$,
\[ \dist\left(\Leb_{\R^n/\Z^n}, \sum_{x\in \Gamma \cap B_{\ell R_0}} \bar\delta_{A'x}\right) < 4\varep.\]
Remark that we can suppose that this remains true on a whole neighbourhood of $A'$. We use the fact that $\Gamma$ is an almost periodic pattern to deduce that $A'$ belongs to the interior of $\mathcal U_\varep$.
\end{proof}

\section{Proof of Theorem \ref{imgquasi}}\label{TrouVentre}

\begin{notation}\label{intelligent}
For $A\in GL_n(\R)$, we denote $A=(a_{i,j})_{i,j}$. We denote by $I_\Q(A)$\index{$I_\Q(A)$} the set of indices $i$ such that $a_{i,j}\in\Q$ for every $j\in\llbracket 1,n\rrbracket$
\end{notation}

The proof of Theorem~\ref{imgquasi} relies on the following remark:

\begin{rem}\label{pmtriv}
If $a\in\Q$, then there exists $q\in\N^*$ such that $\{ax\mid x\in\Z\}\subset \frac{1}{q}\Z$. On the contrary, if $a\in\R\setminus\Q$, then the set $\{ax\mid x\in\Z\}$ is equidistributed in $\R/\Z$.
\end{rem}

Thus, in the rational case, the proof will lie in an argument of periodicity. On the contrary, in the irrational case, the image $A(\Z^n)$ is equidistributed modulo $\Z^n$: on every large enough domain, the density does not move a lot when we perturb the image set $A(\Z^n)$ by small translations. This reasoning is formalized by Lemmas~\ref{tiroir} and \ref{équi}. 

More precisely, for $R$ large enough, we would like to find vectors $w$ such that $D^+_R\big((\pi(A\Gamma) +w)\Delta \pi(A\Gamma)\big)$ is small. We know that there exists vectors $v$ such that $D^+_R\big((\Gamma+v)\Delta\Gamma\big)$ is small; this implies that $D^+_R\big((A\Gamma+Av)\Delta A\Gamma\big)$ is small, thus that $D^+_R\big(\pi(A\Gamma+Av)\Delta \pi(A\Gamma)\big)$ is small. The problem is that in general, we do not have $\pi(A\Gamma+Av) = \pi(A\Gamma)+\pi(Av)$. However, this is true if we have $Av\in\Z^n$. Lemma~\ref{tiroir} shows that in fact, it is possible to suppose that $Av$ ``almost'' belongs to $\Z^n$, and Lemma~\ref{équi} asserts that this property is sufficient to conclude.

The first lemma is a consequence of the pigeonhole principle.

\begin{lemme}\label{tiroir}
Let $\Gamma\subset \Z^n$ be an almost periodic pattern, $\varep>0$, $\delta>0$ and $A\in GL_n(\R)$. Then we can suppose that the elements of $A(\mathcal N_\varep)$ are $\delta$-close to $\Z^n$. More precisely, there exists $R_{\varep,\delta}>0$ and a relatively dense set $\widetilde{\mathcal N}_{\varep,\delta}$\index{$\widetilde{\mathcal N}_{\varep,\delta}$} such that 
\[\forall R\ge R_{\varep,\delta},\  \forall v\in\widetilde{\mathcal N}_{\varep,\delta},\  D_R^+\big( (\Gamma+v)\Delta \Gamma \big) <\varep,\]
and that for every $v\in\widetilde{\mathcal N}_{\varep,\delta}$, we have $d_\infty(Av,\Z^n)<\delta$. Moreover, we can suppose that for every $i\in I_\Q(A)$ and every $v\in\widetilde{\mathcal N}_{\varep,\delta}$, we have $(Av)_i\in \Z$.
\end{lemme}

The second lemma states that in the irrational case, we have continuity of the density under perturbations by translations.

\begin{lemme}\label{équi}
Let $\varep>0$ and $A\in GL_n(\R)$. Then there exists $\delta>0$ and $R_0>0$ such that for all $w\in B_\infty(0,\delta)$ (such that for every $i\in I_\Q(A)$, $w_i=0$), and for all $R\ge R_0$, we have
\[D_R^+\big(\pi(A\Z^n) \Delta \pi(A\Z^n+w) \big) \le \varep.\]
\end{lemme}

\begin{rem}\label{RemContTrans}
When $I_\Q(A) = \emptyset$, and in particular when $A$ is totally irrational (see Definition~\ref{TotIrrat}), the map $v\mapsto \tau(A+v)$ is continuous in 0; the same proof as that of this lemma implies that this function is globally continuous.
\end{rem}

We begin by the proofs of both lemmas, and prove Theorem~\ref{imgquasi} thereafter.

\begin{proof}[Proof of Lemma \ref{tiroir}]
Let us begin by giving the main ideas of the proof of this lemma. For $R_0$ large enough, the set of remainders modulo $\Z^n$ of vectors $Av$, where $v$ is a $\varep$-translation of $\Gamma$ belonging to $B_{R_0}$, is close to the set of remainders modulo $\Z^n$ of vectors $Av$, where $v$ is any $\varep$-translation of $\Gamma$. Moreover (by the pigeonhole principle), there exists an integer $k_0$ such that for each $\varep$-translation $v\in B_{R_0}$, there exists $k\le k_0$ such that $A(k v)$ is close to $\Z^n$. Thus, for every $\varep$-translation $v$ of $\Gamma$, there exists a $(k_0-1)\varep$-translation $v' = (k-1)v$, belonging to $B_{k_0 R_0}$, such that $A(v+v')$ is close to $\Z^n$. The vector $v+v'$ is then a $k_0\varep$-translation of $\Gamma$ (by additivity of the translations) whose distance to $v$ is smaller than $k_0 R_0$.
\medskip

We now formalize these remarks. Let $\Gamma$ be an almost periodic pattern, $\varep>0$ and $A\in GL_n(\R)$. First of all, we apply the pigeonhole principle. We partition the torus $\R^n/\Z^n$ into squares whose sides are smaller than $\delta$; we can suppose that there are at most  $\lceil 1/\delta\rceil^n$ such squares. For $v\in \R^n$, we consider the family of vectors $\{A(kv)\}_{0\le k\le \lceil 1/\delta\rceil^n}$ modulo $\Z^n$. By the pigeonhole principle, at least two of these vectors, say $A(k_1v)$ and $A(k_2v)$, with $k_1<k_2$, lie in the same small square of $\R^n/\Z^n$. Thus, if we set $k_v = k_2-k_1$ and $\ell = \lceil 1/\delta\rceil^n$, we have
\begin{equation}\label{eqdistZ}
1\le k_v\le \ell \quad \text{and} \quad d_\infty\big(A(k_vv),\Z^n\big)\le\delta.
\end{equation}
To obtain the conclusion in the rational case, we suppose in addition that $v\in q\Z^n$, where $q\in\N^*$ is such that for every $i\in I_\Q(A)$ and every $1\le j\le n$, we have $q\, a_{i,j}\in\Z$ (which is possible by Remark~\ref{arithProg2}).

We set $\varep'=\varep/\ell$. By the definition of an almost periodic pattern, there exists $R_{\varep'}>0$ and a relatively dense set ${\mathcal N}_{\varep'}$ such that Equation \eqref{EqAlmPer} holds for the parameter $\varep'$:
\begin{equation}\label{EqAlmPer3}\tag{\ref{EqAlmPer}'}
\forall R\ge R_{\varep'},\  \forall v\in\mathcal N_{\varep'},\  D_R^+\big( (\Gamma+v)\Delta \Gamma \big) <\varep',
\end{equation}

We now set
\[P = \big\{Av\operatorname{mod} \Z^n \mid v\in {\mathcal N}_{\varep'}\big\} \quad \text{and} \quad P_R = \big\{Av\operatorname{mod} \Z^n \mid v\in \mathcal N_{\varep'}\cap B_R\big\}.\]
We have $\bigcup_{R>0} P_R = P$, so there exists $R_0>R_{\varep'}$ such that $d_H(P,P_{R_0})<\delta$ (where $d_H$\index{$d_H$} denotes Hausdorff distance). Thus, for every $v\in\mathcal N_{\varep'}$, there exists $v'\in \mathcal N_{\varep'}\cap B_{R_0}$ such that
\begin{equation}\label{eq666}
d_\infty(Av-Av',\Z^n)<\delta.
\end{equation}

We then remark that for every $v'\in {\mathcal N}_{\varep'}\cap B_{R_0}$, if we set $v'' = (k_{v'}-1)v'$, then by Equation \eqref{eqdistZ}, we have
\[d_\infty(Av' + Av'',\Z^n) = d_\infty\big(A(k_{v'}v'),\Z^n\big)\le\delta.\]
Combining this with Equation~\eqref{eq666}, we get
\[d_\infty(Av + Av'',\Z^n)\le 2\delta,\]
with $v''\in B_{\ell R_0}$. 

On the other hand, $k_{v'}\le \ell$ and Equation \eqref{EqAlmPer3} holds, so Lemma \ref{arithProg} (more precisely, Remark~\ref{arithProg2}) implies that $v''\in \mathcal N_\varep$, that is
\[\forall R\ge R_{\varep'},\  D_{R}^+\big( (\Gamma+ v'')\Delta \Gamma \big) <\varep.\]

In other words, for every $v\in\mathcal N_{\varep'}$, there exists $v''\in \mathcal N_\varep \cap B_{\ell R_0}$ (with $\ell$ and $R_0$ independent from $v$) such that $d_\infty\big(A(v+v''),\Z^n\big)<2\delta$. The set $\widetilde{\mathcal N}_{2\varep,2\delta}$ we look for is then the set of such sums $v+v''$.
\end{proof}

\begin{proof}[Proof of Lemma \ref{équi}]
Under the hypothesis of the lemma, for every $i\notin I_\Q(A)$, the sets
\[\left\{\sum_{j=1}^n a_{i,j} x_j\mid (x_j)\in\Z^n\right\},\]
are equidistributed modulo $\Z$. Thus, for all $\varep>0$, there exists $R_0>0$ such that for every $R\ge R_0$,
\[D_R^+\big\{v\in\Z^n \,\big|\, \exists i\notin I_\Q(A) : d\big((Av)_i,\Z+\frac12\big)\le \varep\big\} \le 2(n+1)\varep.\]
As a consequence, for all $w\in\R^n$ such that $\|w\|_\infty\le\varep/(2(n+1))$ and that $w_i=0$ for every $i\in I_\Q(A)$, we have
\[D_R^+\big(\pi(A\Z^n) \Delta \pi(A(\Z^n+w))\big)\le\varep.\]
Then, the lemma follows from the fact that there exists $\delta>0$ such that $\|A(w)\|_\infty\le \varep/(2(n+1))$ as soon as $\|w\|\le\delta$.
\end{proof}

\begin{proof}[Proof of Theorem \ref{imgquasi}]
Let $\varep>0$. Lemma \ref{équi} gives us a corresponding $\delta>0$, that we use to apply Lemma \ref{tiroir} and get a set of translations $\widetilde{\mathcal N}_{\varep,\delta}$. Then, for every $v\in \widetilde{\mathcal N}_{\varep,\delta}$, we write $\pi(Av) = Av + \big(\pi(Av)-Av\big) = Av + w$. The conclusions of Lemma~\ref{tiroir} imply that $\|w\|_\infty <\delta$, and that $w_i=0$ for every $i\in I_\Q(A)$.

We now explain why $\hat Av = \pi(Av)$ is a $\varep$-translation for the set $\widehat A(\Gamma)$. Indeed, for every $R\ge \max(R_{\varep,\delta},MR_0)$, where $M$ is the maximum of the greatest modulus of the eigenvalues of $A$ and of the greatest modulus of the eigenvalues of $A^{-1}$, we have
\begin{align*}
D^+_R \Big(\pi(A\Gamma) \Delta \big(\pi(A \Gamma)+\widehat Av\big)\Big) \le &\ D^+_R \Big(\pi(A\Gamma) \Delta \big(\pi(A \Gamma)+w\big)\Big)\\
                  & + D^+_R \Big(\big(\pi(A\Gamma) + w\big) \Delta \big(\pi(A \Gamma)+\widehat Av\big)\Big)
\end{align*}
(where $w=\pi(Av)-Av$). By Lemma \ref{équi}, the first term is smaller than $\varep$. For its part, the second term is smaller than
\[D^+_R\big((A\Gamma + Av) \Delta A \Gamma\big) \le M^2 D^+_{RM}\big((\Gamma + v) \Delta \Gamma\big),\]
which is smaller than $\varep$ because $v\in\mathcal N_\varep$.
\end{proof}

\bibliographystyle{amsalpha}
\bibliography{../../Biblio}

\end{document}